\documentclass{article}
\usepackage[margin=1in]{geometry}

\usepackage{amsmath,amssymb,amsthm}
\usepackage{tikz,enumerate}
\usepackage{mathtools}
\usepackage{etex,mathrsfs}
\usepackage{tikz}  
\usetikzlibrary{positioning} 
\usetikzlibrary{calc}
\usetikzlibrary{decorations.markings}

\theoremstyle{plain}

\newtheorem{theorem}{Theorem}
\newtheorem{proposition}{Proposition}

\newtheorem{lemma}{Lemma}

\theoremstyle{remark}

\newtheorem{remark}{Remark}

\theoremstyle{definition}
\newtheorem{definition}{Definition}

\usepackage{cite}
\usepackage{amsmath,amssymb,amsfonts}
\usepackage{algorithmic}
\usepackage{graphicx}
\usepackage{textcomp}
\usepackage{xcolor}
\usepackage{tikz,subcaption}
\usetikzlibrary{positioning,fit,calc,arrows.meta}

\def\BibTeX{{\rm B\kern-.05em{\sc i\kern-.025em b}\kern-.08em
    T\kern-.1667em\lower.7ex\hbox{E}\kern-.125emX}}
\begin{document}

\title{Stable Source Coding
}

\author{ Zhenduo Wen and Amin Gohari}

\maketitle

\begin{abstract}

A source encoder is stable if a small change in the source sequence (e.g., changing a few symbols) results in a small (or bounded) change in the output codeword. By this definition, the common technique of random binning is unstable; because the mapping is random, two nearly identical source sequences can be assigned to completely unrelated bin indices. We study compression rates of stable lossless source codes.  Using combinatorial arguments, we derive information-theoretic limits on the achievable rate as a function of the stability parameters. 

\end{abstract}

\section{Introduction}

Artificial neural-network-based (ANN-based) compressors are increasingly being adopted in modern source coding and multimedia compression (e.g., \cite{wagner2021neural, balle2020nonlinear, AyferChen2025information, ozyilkan2024distributed, yang2023introduction} and the references therein). Beyond their empirical success, a recurring observation is that networks trained by gradient-based methods often exhibit limited sensitivity
to small input perturbations. 
This behavior is closely connected to notions of \emph{algorithmic stability}
\cite{bousquet2002stability,hardt2016train,elisseeff2005stability,kutin2002almost},
where small changes in the input lead to small changes in the output.

While such stability is desirable for robustness and generalization, it can be at odds with classical source-coding constructions,
which rely on highly discontinuous partitions of the source space (e.g., quantization and binning). In random binning, even source sequences that are very close to one another can be assigned to completely unrelated bins. A similar issue arises in lossy source coding using vector quantization; the mapping is discontinuous at the boundaries of quantization cells, meaning a tiny shift in the input can trigger a sudden switch to a different quantization vector.

Recent empirical and theoretical work has reported that ANN-based compressors may underperform on sources or tasks
where optimal compression requires learning discontinuous encoders/partitions
\cite{bhadane2022neural,ozyilkan2024neural,ozyilkan2023neural}.
These findings suggest that Shannon-style achievability results, which allow arbitrary encoders,
may not fully predict the behavior of learned compressors operating under stability-like constraints.

Motivated by this gap, we study lossless source coding under an explicit stability requirement: An encoder $f_n:\mathcal{X}^n\rightarrow \{0,1\}^{nR}$ is said to be \emph{$(D_n,D_n')$-stable} if for all
$x^n,\tilde x^n\in\mathcal{X}^n$,
\begin{equation}
    d_I(x^n,\tilde x^n)\le D_n
\implies
d_H\bigl(f_n(x^n),f_n(\tilde x^n)\bigr)\le D_n' \label{eq:stab_on_A}
\end{equation}
where $d_H$ is the Hamming distance on binary sequences, and $d_I(x^n,\tilde x^n)=\sum_{i=1}^nd_I(x_i,\tilde x_i)$
for some distance metric $d_I:\mathcal{X}\times\mathcal{X}\rightarrow \mathbb{R}_+$.
Our goal is to characterize the minimum achievable rate $R$ as a function of $\{(D_n,D_n')\}_{n \geq 1}$.

We note that there is a different notion of stability used in \emph{universal source coding} \cite{ziv2003universal}, in which stability refers to the algorithm's ability to perform well (achieve a compression rate close to entropy) even when the source statistics are not perfectly stationary or known.

Our notion of stability can be viewed as a local version of \emph{metric embedding}: only pairs within distance $D_n$ must remain within distance $D_n'$ after mapping.
This is weaker than global distance-preserving embeddings studied in the metric embedding literature and in Johnson--Lindenstrauss-type results
\cite{johnson1984extensions,dirksen2024fast,plan2014dimension}. Moreover, \emph{minimum-distortion embedding} approaches \cite{agrawal2021minimum} are algorithmic in nature and can give insights on achievability,
whereas our focus here is on converse bounds for achievable rates under stability constraints.

\smallskip
\noindent\textbf{Contributions.}
We develop converse bounds by reducing stability-constrained lossless coding to a graph homomorphism problem.
After restricting to a fixed type class,
we view sequences as vertices of a certain union of \emph{generalized Johnson graphs} on the source side,
and codewords as vertices of a \emph{Hamming graph} on $\{0,1\}^{nR}$.
The stability constraint implies that, on a large induced subgraph corresponding to correct decoding, the encoder induces an \emph{injective graph homomorphism} between these graphs.
What such a homomorphism can preserve yields explicit necessary conditions on $(R,D_n,D_n')$.

\smallskip
\noindent\textbf{Organization.}

Section \ref{sec::prelims} introduces the formal problem setup and necessary notations. Section \ref{sec::converse_bounds} states the converse bounds we derived. Section \ref{sec::proofs} presents the proof of our converse bounds. Section \ref{sec::visualizations} and \ref{sec::discussions} discuss numerical evaluations and possible future directions.

\section{Preliminaries}\label{sec::prelims}
\subsection{Problem Setup}\label{sec::problem_setup}

Let $\{X_i\}_{i=1}^\infty$ be an independent and identically distributed (i.i.d.)\ source with discrete alphabet $\mathcal{X}$ and distribution $P_X$. Without loss of generality, assume that $P_X(x)>0$ for all $x$.
Given a compression rate $R$ and for blocklength $n \in \mathbb{N}$, a (lossless) compression scheme consists of:
\begin{itemize}
  \item An encoder $
    f_n: \mathcal{X}^n \to \{0,1\}^{nR},
  $
  that maps a length-$n$ source sequence $X^n = (X_1,\dots,X_n)$ to a binary string of length $nR$ bits.
  \item A decoder 
  $
    g_n: \{0,1\}^{nR} \to \mathcal{X}^n,
  $
  that maps the binary description back to an estimate $\hat{X}^n = g_n(f_n(X^n))$ of the original sequence.
\end{itemize}

 Given sequences $\{D_n\}$ and $\{D'_n\}$, we say that a rate $R$ is $(D_n,D'_n)$-\emph{achievable} if there is a sequence of lossless codes $\{(f_n,g_n)\}_{n\ge1}$ such that the decoding error probability
\begin{equation}
      P_e^{(n)} \triangleq \Pr\big[g_n(f_n(X^n)) \neq X^n\big] \to 0 \quad \text{as } n \to \infty,\label{def::error_prob}
\end{equation}
and the encoding function $f_n$ satisfies \eqref{eq:stab_on_A} for all $n$.

\subsection{Notations}
The function
$h_2(x) \triangleq-x\log_2 x-(1-x)\log_2(1-x)$ denotes the binary entropy (base-$2$). We use $\log(\cdot)$ to denote logarithms in base~$2$, and $\ln(\cdot)$ to denote the natural logarithm. $\Delta(G)$ and $\omega(G)$ denote the maximum degree and maximum clique size of graph $G$. $d_H(\cdot, \cdot)$ is the Hamming distance.

For $x^n\in\mathcal X^n$, let
\begin{equation}
    N(a|x^n)\triangleq \sum_{i=1}^n \mathbf 1\{x_i=a\},\qquad
\hat P_{x^n}(a)\triangleq \frac{1}{n}N(a|x^n)
\end{equation}
be the empirical counts and empirical frequencies. 
Let $\Pi_n$ denote the set of all types on $\mathcal X$ with denominator $n$,
and for $\pi_n\in\Pi_n$ let the type class be
\begin{equation}
    \mathcal T(\pi_n)\triangleq \{x^n\in\mathcal X^n:\hat P_{x^n}=\pi_n\}.
\end{equation}

Let $\Pi_n({\epsilon}, P_X)$ be the set of all ${\epsilon}$-strongly typical types,
\begin{equation}
    \Pi_n(\epsilon, P_X)\triangleq\Bigl\{\pi_n\in\Pi_n:\forall a\in\mathcal X,\ |\pi_n(a)-P_X(a)|\le \epsilon\Bigr\}.
\end{equation}

For $\epsilon>0$, define the $\epsilon$-strongly typical set
\begin{equation}
    \mathcal T_\epsilon^{(n)}(X)\triangleq \Bigl\{x^n\in\mathcal X^n:
\forall a\in\mathcal X,\ |\hat P_{x^n}(a)-P_X(a)|\le \epsilon\Bigr\}.
\end{equation}

Consider a sequence $\epsilon_n$ such that $\epsilon_n\rightarrow 0$ and $\epsilon_n\sqrt{n}\rightarrow\infty$ as $n$ goes to infinity. Then, the probability of the typical set approaches one under an i.i.d. source \cite{csiszar2011information}.

\section{Converse bounds}\label{sec::converse_bounds}

To keep the notation light and to highlight the main ideas, we present the construction
and the subsequent degree/clique computations in the binary case $\mathcal X=\{0,1\}$. Moreover, we use per-letter Hamming distance, so $d_I=d_H$ on $\{0,1\}$ and the induced block distance is the usual Hamming distance. Without loss of generality, we assume that $0 < p \triangleq P_X(1) < \frac12$.

For sequences $\{(D_n, D'_n)\}_{n\ge1}$, we consider two regimes:
\begin{itemize}
\item[(i.)] \textit{Linear regime:} $D_n=\Theta(n)$ and $D_n'=\Theta(n)$. Here, we assume the limits
\[
d_1 \triangleq \lim_{n\to\infty}\frac{D_n}{n} > 0,\qquad
d_2 \triangleq \lim_{n\to\infty}\frac{D_n'}{n} > 0
\]
exist with $0\le \frac{d_1}{2} \le p$ and $0 \leq d_2 \leq R$. The reason for this assumption is as follows: intuitively speaking, any two typical sequences $x_1^n$ and $x_2^n$ have weight $np$ and Hamming distance at most $2np$. Therefore, $d_1>2p$ is a trivial case. Similarly, we assume $0 \leq d_2 \leq R$.
\item[(ii.)] \textit{Sublinear regime:} $D_n=o(n)$ and $D_n'=o(n)$.
\end{itemize}

Below, we state our converse bounds on $(R, D_n, D_n')$. The proofs can be found in section \ref{sec::proofs}.

\begin{theorem}[Degree bound in the linear regime]\label{thm::degree_linear} Define
\begin{align*}
&s(x)\triangleq p\,h_2\!\Big(\frac{x}{p}\Big)+(1-p)\,h_2\!\Big(\frac{x}{1-p}\Big),
\\&F(d_1, p)\triangleq 
\begin{cases}
    s(\frac{d_1}{2}), &\quad \ 0 \leq \frac{d_1}{2} \leq p(1-p),
    \\h_2(p), &\quad \ p(1-p)\leq \frac{d_1}{2},
\end{cases}
\\&G(d_2, R)\triangleq
\begin{cases}
R\cdot h_2(\frac{d_2}{R}), & 0\le d_2 < R/2,\\
R, & d_2 \ge R/2.
\end{cases}
\end{align*}
Then, for all  \emph{$(nd_1,nd_2)$-achievable rates $R$}, the following must hold: $
F(d_1, p)\ \le\ G(d_2, R).$

\end{theorem}

\begin{theorem}[Clique size bound in the linear regime]\label{thm::clique_size_linear}
     Define
\begin{align*}
&\Psi(d_2,R)\triangleq
\begin{cases}
R\cdot h_2\!\big(\frac{d_2}{2R}\big), & 0\le d_2< R,\\[3pt]
R, & d_2\ge R.
\end{cases}
\end{align*}
Then for all  \emph{$(nd_1,nd_2)$-achievable rates $R$} where
$
\frac{d_1}{2}< p
$ the following must hold:
$
h_2(d_1/2)\ \le\ \Psi(d_2,R).$

\end{theorem}

\begin{proposition}[Sublinear regime]\label{prop::sub_linear}
    Given a binary i.i.d source $P_X$, $D_n = o(n)$ and $D_n' = o(n)$,
for all sequences of $(D_n,D_n')$-stably decodable codes:
\begin{equation}
\lim_{n\rightarrow\infty}n^{D_n-D'_n}\cdot \frac{D_n'^{D_n'}\left((p-\epsilon_n)(1-p-\epsilon_n)\right)^{D_n/2}}{(D_n/2)^{D_n}\left(e^2R\right)^{D_n'}}\leq 1
\end{equation}
for any sequence $\epsilon_n$ such that $\epsilon_n\rightarrow 0$ and $\epsilon_n\sqrt{n}\rightarrow\infty$ as $n$ goes to infinity.
 As a corollary, if $(D_n,D_n')=(D,D')$ for some constants $D$ and $D'$, then, we must have
$    D \leq D'$. This is true regardless of the rate $R$, and is nontrivial, especially when $R > 1$, in which case the length expands rather than compresses.

\end{proposition}

\section{Proofs}\label{sec::proofs}

\subsection{Reduction to one type class}\label{sec::reduction_one_class}

\begin{proposition}[Reduction to one type class]\label{prop:reduction-one-type}
Take a sequence $\epsilon_n$ such that $\epsilon_n\rightarrow 0$ and $\epsilon_n\sqrt{n}\rightarrow\infty$ as $n$ goes to infinity. Consider a sequence $(f_n,g_n)$ of $(D_n,D_n')$-codes with rate $R$.
Then there exists a sequence of types $\{\pi_n:\pi_n\in\Pi_n({\epsilon_n},P_X)\}$ and a sequence of sets
$\mathcal A_n\subset \mathcal T(\pi_n)$ such that
\begin{equation}
\lim_{n\to\infty}\frac{|\mathcal A_n|}{|\mathcal T(\pi_n)|}=1\label{defrationeq}
\end{equation}
and for all $x^n\in\mathcal A_n$ we have $g_n(f_n(x^n))=x^n$.

\end{proposition}
The proof of Proposition \ref{prop:reduction-one-type} can be found in the appendix.
\addtolength{\topmargin}{0.05in}
\subsection{Source and codeword graphs and their homomorphism}

In the binary setting, a fixed type class $\mathcal T(\pi_n)$ is a constant-weight slice of the Hamming cube. 

\begin{definition}[Source neighborhood graph]\label{def:Gn}

Let $G_n=G_n(\pi_n,D_n)$ be the graph with vertex set $V(G_n)=\mathcal T(\pi_n)$ and
an edge between two distinct vertices $x^n,\tilde x^n\in\mathcal T(\pi_n)$ iff
$
d_H(x^n,\tilde x^n)\le D_n.
$
\end{definition}
Note that the above graph is a union of certain \emph{generalized Johnson graphs}, also known as \emph{uniform subset graphs} \cite{chen1987hamiltonian}. Note that $G_n$ is vertex-transitive.

\begin{definition}[Codeword neighborhood graph]\label{def:Hn}
Let $\ell_n\triangleq \lfloor nR\rfloor$. Let $H_n=H_n(\ell_n,D_n')$ be the graph with vertex set $V(H_n)=\{0,1\}^{\ell_n}$ and
an edge between two distinct vertices $u,v\in\{0,1\}^{\ell_n}$ iff
$
d_H(u,v)\le D_n'.
$
Equivalently, $H_n(\ell_n,D_n')$ is the $D_n'$-th power of the $\ell_n$-dimensional
binary Hamming graph.
\end{definition}

Consider a sequence of types $\{\pi_n:\pi_n\in\Pi_n({\epsilon_n},P_X)\}$ and a sequence of sets
$\mathcal A_n\subset \mathcal T(\pi_n)$ as in Proposition~\ref{prop:reduction-one-type}. Let $G_n'$ denote the induced subgraph of $G_n$ on $\mathcal A_n$. Then, the fact that $g_n(f_n(x^n))=x^n$ on $\mathcal A_n$, shows that 
there is an \emph{injective graph homomorphism} from $G_n'$ to $H_n$, defined by $f_n$: note that the stability condition yields
\begin{align*}
    &(x^n,\tilde x^n)\in E(G_n') 
     \implies  (f_n|_{\mathcal A_n}(x^n),f_n|_{\mathcal A_n}(\tilde x^n))\in E(H_n),
\end{align*}
and $f_n|_{\mathcal A_n}$ is one-to-one on $V(G_n')$ due to $g_n(f_n(x^n))=x^n$.

From the graph homomorphism property, we deduce
\begin{align}
    \Delta(G'_n)&\leq \Delta(H_n),\label{eqnDDD1}\\
     \omega(G'_n)&\leq \omega(H_n)\label{eqnDDD2}
\end{align}
where $\Delta(\cdot)$ and $\omega(\cdot)$ are the maximum degree and maximum clique size.

\subsection{Relating $\Delta(G'_n)$ and $\omega(G'_n)$ to $\Delta(G_n)$ and $\omega(G_n)$}

The following lemma shows that the induced subgraph $G_n'$ inherits (almost) the maximum degree and clique number from $G_n$. Note that for fixed $\pi_n$ and $D_n$, the graph $G_n(\pi_n,D_n)$ is vertex-transitive.

\begin{lemma}[Maximum degree and maximum clique size of a large induced subgraph]\label{lem:deg-preserve-simple}
Let $G$ be a vertex-transitive graph on $N$ vertices of a (regular) degree $d$. Let $S\subseteq V(G)$ have size $M$ and $G[S]$ be the induced subgraph on $S$. Then its maximum degree $\Delta(G[S])$ and maximum clique size $\omega(G[S])$ satisfy:
\begin{equation}
    \Delta(G[S])\ 
\ge\ d\Bigl(2-\frac{N}{M}\Bigr),  \quad (M \geq \frac{N}{2})
\end{equation}
\begin{equation}
    \omega(G[S])\ \ge\ \Bigl\lceil \frac{M}{N}\,\omega(G)\Bigr\rceil.
\end{equation}
In particular, if $\frac{M}{N}\ge 1-\varepsilon$, then
\begin{equation}
    \Delta(G[S])\ \ge\ d\cdot \frac{1-2\varepsilon}{1-\varepsilon}, \quad 
    \omega(G[S])\ge \lceil (1-\varepsilon)\cdot \omega(G)\rceil.
\end{equation}
\end{lemma}

We leave the proof of this lemma to the appendix.

From \eqref{defrationeq}, we have that
\begin{align}
\lim_{n\rightarrow\infty}\frac{|V(G'_n)|}{|V(G_n)|}=1.
\end{align}
Thus, Lemma \ref{lem:deg-preserve-simple} along with \eqref{eqnDDD1} and \eqref{eqnDDD2} imply that
\begin{align}
\limsup_{n\rightarrow\infty}\frac{\Delta(G_n)}{\Delta(H_n)}\leq 1,\label{eqnDeltaBound}
\\
\limsup_{n\rightarrow\infty}\frac{\omega(G_n)}{\omega(H_n)}\leq 1.\label{eqnOmegaBound}
\end{align}

\subsection{A maximum degree converse bound}\label{subsec:degree-converse}

Assume that $p \triangleq P_X(1) < \frac12$. 
Take a sequence $\epsilon_n$ such that $\epsilon_n\rightarrow 0$ and $\epsilon_n\sqrt{n}\rightarrow\infty$ as $n$ goes to infinity. Proposition \ref{prop:reduction-one-type} yields a sequence of types $\pi_n$ within ${\epsilon_n}$ distance of $P_X$. Let
$p_n \triangleq \pi_n(1)$.
Then, $\lim_{n\rightarrow\infty}p_n=p$. Let $k_n \triangleq np_n \in \mathbb Z$. Assume $\lfloor D_n/2 \rfloor \leq k_n$. Then, 
\begin{align}
\Delta(G_n)&=\sum_{j=1}^{\lfloor D_n/2\rfloor}\binom{k_n}{j}\binom{n-k_n}{j}
\geq\max_{1\leq j\leq \lfloor D_n/2\rfloor}\binom{k_n}{j}\binom{n-k_n}{j}\label{eq:degG-exact}
\end{align}
because two vertices in $G_n$ are adjacent iff they differ in at most $D_n$ coordinates, i.e.,
they swap at most $j\le \lfloor D_n/2\rfloor$ ones with zeros.

Likewise, $H_n$ is the $D_n'$-th power of the $\ell_n$-cube where $\ell_n= \lfloor nR\rfloor$, so
\begin{align}
\Delta(H_n)&=\sum_{i=1}^{\min(D_n',\ \ell_n)}\binom{\ell_n}{i}
\label{eq:degH-exact}
\end{align}
Note that the index starts from $1$ because edges only connect {distinct} vertices. 

It is known that for $0 \leq k\leq n$, we have
\begin{align}
    \binom{n}{k}&\geq \frac{1}{n+1}2^{nh_2(k/n)}\label{eqnEE1}\\
    \binom{n}{k}&\geq \frac{n^k}{k^k}\label{eqnEE2}
\end{align}

It is also known that for $0 \leq k \leq n/2$, we have
\begin{align}
    \sum_{i=0}^k \binom{n}{i}&\leq 2^{n h_2(k/n)},\label{expnup1}\\
 \sum_{i=0}^k \binom{n}{i}&\leq (k+1)\binom{n}{k}\leq (k+1) \frac{(en)^k}{k^k}.\label{expnup2} 
\end{align}

\begin{proof}[Proof of Theorem \ref{thm::degree_linear}]

\textit{(i.) Linear regime.}

From \eqref{eqnDeltaBound}:
\begin{align}
    &\log \Big [\limsup_{n\rightarrow\infty}\frac{\Delta(G_n)}{\Delta(H_n)}\Big]\leq 0 
    \\& \implies \limsup_{n\rightarrow\infty}\Big [\log {\Delta(G_n)} - \log {\Delta(H_n)} \Big] \leq 0 \label{eq:deg_logranged}
    \\& \implies \limsup_{n\rightarrow\infty}\Big [ \frac1n\log {\Delta(G_n)} - \frac1n\log {\Delta(H_n)} \Big] \leq 0 \notag
    \\& \implies \limsup_{n\rightarrow\infty}\Big [\frac1n\log {\Delta(G_n)}\Big ] \leq \limsup_{n\rightarrow\infty}
     \Big [\frac1n\log {\Delta(H_n)} \Big] .\label{eq:degree_degenerated}
\end{align}

We know that \(\lim_{n \to \infty} D_n/n = d_1\), \(\lim_{n \to \infty} D_n'/n = d_2\), $\lim_{n \to \infty}k_n/n =p$ and $\lim_{n \to \infty}\ell_n/n = R$. In the linear regime, $0\le \frac{d_1}{2} \le p$ and $0 \leq d_2 \leq R$.

If $D'_n<\ell_n/2$, from \eqref{expnup1}  we have $\Delta(H_n)\leq 2^{\ell_nh_2({D'_n}/\ell_n)}$.

If $D'_n>\ell_n/2$, we have $\Delta(H_n)\leq 2^{\ell_n}$.

Thus, we can write both cases together as
$\limsup_{n \to \infty} \frac{1}{n} \log \Delta (H_n)\leq 
    \max_{0 \leq y \leq d_2} R\cdot h_2(\frac{y}{R})$.

From \eqref{eq:degG-exact} and \eqref{eqnEE1}, we have 
\begin{align*}
    \limsup_{n \to \infty} \frac{1}{n} \log \Delta (G_n) \geq \max_{0\le x\le \frac{d_1}{2}} s(x)
\end{align*}
where $s(x)\triangleq p\,h_2\!\Big(\frac{x}{p}\Big)+(1-p)\,h_2\!\Big(\frac{x}{1-p}\Big).$
Thus, from \eqref{eq:degree_degenerated} we obtain,
\begin{align*}
\max_{0\le x\le \frac{d_1}{2}} s(x) \leq \max_{0 \leq y \leq d_2} R\cdot h_2(\frac{y}{R}).
\end{align*}

Observe that, $s(p(1-p)) = h_2(p)$. Recall $\Delta(G_n) \leq |G_n|$, hence $\lim_{n \to \infty} \frac{1}{n}\log \Delta(G_n) \leq \lim_{n \to \infty} \frac{1}{n} \log |G_n| = h_2(p)$. Thus, $\max_{0\le x\le {d_1/2}} s(x) = h_2(p)$ if $p(1-p) \leq \frac{d_1}{2}$. Moreover, $s(\cdot)$ is concave. We can show that $s'(x) = \ln (\frac{p-x}{x}) + \ln(\frac{1-p-x}{x}) \geq 0$ when $0 \leq x \leq p(1-p)$. Hence, $\sup_{0\le x\le {d_1/2}} s(x) = s(\frac{d_1}{2})$ if $\frac{d_1}{2}\leq p(1-p)$. Using a similar argument we can show that
\begin{align*}
\max_{0 \leq y \leq d_2} \left [R\cdot h_2(\frac{y}{R}) \right]=
\begin{cases}
R\cdot h_2(\frac{d_2}{R}), & 0\le d_2 < R/2,\\
R, & d_2 \ge R/2.
\end{cases}
\end{align*}

\smallskip
\noindent\textit{(ii.) Sublinear regime.}
 
Using \eqref{eq:degG-exact} and \eqref{eqnEE2}, we have the following for some constant $C_1$:
\begin{align}
\Delta(G_n)&\geq\max_{1\leq j\leq \lfloor D_n/2\rfloor}\left(\frac{k_n(n-k_n)}{j^2}\right)^j
\\&\geq \left(\frac{k_n(n-k_n)}{(D_n/2)^2}\right)^{D_n/2}
\\&\geq \left(\frac{n^2(p-\epsilon_n)(1-p-\epsilon_n)}{(D_n/2)^2}\right)^{D_n/2}.\label{eq:deg_sublinear_1}
\end{align}
Since $D_n'=o(n)$, we expect $D_n'\leq \ell_n/2$ to hold for large enough $n$. For clean asymptotics, w.l.o.g take $D_n$ even. Using \eqref{eq:degH-exact} and \eqref{expnup2},  we have the following for $D_n'\leq \ell_n/2$
\begin{align}
\Delta(H_n)&\leq(1+D_n')\left(\frac{enR}{D_n'}\right)^{D_n'}  \leq\left(\frac{e^2nR}{D_n'}\right)^{D_n'}, \label{eq:deg_sublinear_2}
\end{align}
where we used $(1+D_n')^{1/D_n'} \leq e$. 
Thus, we obtain that 
\begin{align*}
    \lim_{n\rightarrow\infty}n^{D_n-D'_n}\cdot \frac{D_n'^{D_n'}\left((p-\epsilon_n)(1-p-\epsilon_n)\right)^{D_n/2}}{(D_n/2)^{D_n}\left(e^2R\right)^{D_n'}}\leq 1.
\end{align*}

If $D_n=D$ and $D'_n=D'$ are constant, we obtain that we must have $D\leq D'$ regardless of the rate $R$.

\end{proof}

\subsection{A clique size converse bound}\label{sec:clique-converse}
Let us recall the complete intersection theorem conjectured in \cite{erd6s1961intersection} and proved by the Ahlswede-Khachatrian theorem   \cite{ahlswede1997complete}. 
Let $n, k, t$ be integers such that $1 \le t \le k \le n$. Let $\binom{[n]}{k}$ denote the set of all subsets of $\{1, \dots, n\}$ of size $k$. A family $\mathcal{F} \subseteq \binom{[n]}{k}$ is called \textit{$t$-intersecting} if $|A \cap B| \ge t$ for all $A, B \in \mathcal{F}$.

For any integer $r \ge 0$, we define the family $\mathcal{F}_r$ as:
\begin{equation}
    \mathcal{F}_r = \left\{ A \in \binom{[n]}{k} : |A \cap [t+2r]| \ge t+r \right\}.
\end{equation}
The size of this family is given by:
\begin{equation}
    |\mathcal{F}_r| = \sum_{j=t+r}^{\min(k, t+2r)} \binom{t+2r}{j} \binom{n-(t+2r)}{k-j}.
\end{equation}

\begin{theorem}[Ahlswede \& Khachatrian \cite{ahlswede1997complete}, 1997]\label{thm:Ahlswede}
    The maximum size of a $t$-intersecting family of $k$-subsets of an $n$-set is
    \begin{equation}
        M(n, k, t) = \max_{r \in \{0, 1, \dots, \lfloor(n-t)/2\rfloor \}} |\mathcal{F}_r|.
    \end{equation}
    Furthermore, the maximum is attained by the family $\mathcal{F}_r$ with
    \begin{equation}
        r=\left\lceil \frac{(k-t+1)(t-1)}{n-2(k-t+1)}-1\right\rceil. \label{eq:r_max}
    \end{equation}
\end{theorem}

In our problem setting, consider $G_n(\pi_n,D_n)$ for some $\pi_n$ and $D_n$. Let $k_n=n\pi_n(1)$. Then, 
\begin{align}
    \omega(G_n)&=M\left(n,k_n,\left\lceil k_n-\frac{D_n}{2}\right\rceil\right)
    \\&\geq\max_{r_n,j_n}\binom{t_n+2r_n}{j_n} \binom{n-(t_n+2r_n)}{k_n-j_n}\label{eqnLBs}
\end{align}
where $t_n=\left\lceil k_n-\frac{D_n}{2}\right\rceil = k_n - \lfloor \frac{D_n}{2} \rfloor$, $r_n \in \{0, 1, \dots, \lfloor(n-t_n)/2\rfloor \}$, and $j_n \in \{t_n+r_n, \dots, \min(k_n, t_n+2r_n) \}$.

On the other hand, a clique in $H(\ell_n, D'_n)$ is exactly a subset of the $\ell_n$-dimensional hypercube with
{diameter} at most Hamming distance $D'_n$. The maximum size
of such a subset is (also) conjectured in \cite{erd6s1961intersection} and is then proved by Kleitman in \cite{kleitman1966combinatorial}.

\begin{theorem}[Kleitman \cite{kleitman1966combinatorial}, 1966]\label{thm:kleitman}We have
    \begin{equation*}
\omega(H(\ell_n,D'_n)) = \begin{cases}
2^{\ell_n}, & \!\!\!\!D'_n\ge \ell_n,\\
B(\ell_n,v_n), & \!\!\!\!D'_n=2v_n<\ell_n,\\
B(\ell_n,v_n)\!+\!\binom{\ell_n-1}{v_n}, & \!\!\!\!D'_n=2v_n+1<\ell_n.
\end{cases}
\end{equation*}
\end{theorem}
where for $v_n,\ell_n \in \mathbb N$, $\ 0\le v_n\le \ell_n$, and $B(\ell_n,v_n)\triangleq \sum_{i=0}^{v_n}\binom{\ell_n}{i}. $

{

\begin{proof}[Proof of Theorem \ref{thm::clique_size_linear}]

\textit{(In Linear regime.)} 

Recall that in the linear regime,\(\lim_{n \to \infty} D_n/n = d_1\), \(\lim_{n \to \infty} D_n'/n = d_2\), $\lim_{n \to \infty}k_n/n =p$ and $\lim_{n \to \infty}\ell_n/n = R$. Moreover, $0\le \frac{d_1}{2} \le p$ and $0 \leq d_2 \leq R$.

Assuming $r_n = n \alpha$ and $j_n = n \beta$.  From \eqref{eqnLBs} and \eqref{eqnEE1},
\begin{align}
    \limsup_{n \to \infty} \frac1n \log 
    \omega(G_n) \geq \Lambda(d_1,p) \label{ineq:log_clique_lower}
\end{align}

where
\begin{align*}
\Lambda(d_1,p)
&\triangleq \max_{\alpha \in \mathcal S_1, \beta \in \mathcal S_2(\alpha)}
(p-\frac{d_1}{2} + 2\alpha) h_2\left( \frac{\beta}{p - \frac{d_1}{2} + 2\alpha } \right) 
\\& \qquad\qquad\qquad + (1-p+\frac{d_1}{2} - 2\alpha)h_2 \left ( \frac{p-\beta}{1-p+\frac{d_1}{2} - 2\alpha} \right),
\\ \mathcal S_1 &= \left\lbrace \alpha:0 \leq \alpha \leq \frac{1-p+\frac{d_1}{2}}{2}, \mathcal S_2(\alpha) \neq \emptyset \right\rbrace,
\\ \mathcal S_2(\alpha) &= \left \lbrace \beta: p-\frac{d_1}{2}+\alpha \leq \beta \leq \min\left[p, p - \frac{d_1}{2} + 2\alpha\right] \right \rbrace.
\end{align*}

On the other hand, from Theorem \ref{thm:kleitman} and \eqref{expnup1}, we have

\begin{align}
    \limsup_{n \to \infty} \frac1n \log \omega(H_n) \leq \max_{0 \leq \gamma \leq \frac{d_2}{2}} R\cdot h_2(\frac{\gamma}{R})= \Psi(d_2, R) \label{ineq:log_clique_upper}
\end{align}
where 
\begin{align*}
    \Psi(d_2,R)&\triangleq
\begin{cases}
R\cdot h_2\!\big(\frac{d_2}{2R}\big), & 0\le d_2< R,\\[3pt]
R, & d_2\ge R.
\end{cases}
\end{align*}

From \eqref{eqnOmegaBound}, we have 
\begin{equation}
    \limsup_{n \to \infty} \frac1n \log \omega(G_n) \leq \limsup_{n \to \infty} \frac1n \log \omega(H_n). \label{ineq:log_clique_ess}
\end{equation}

Combine \eqref{ineq:log_clique_lower}, \eqref{ineq:log_clique_upper}, and \eqref{ineq:log_clique_ess}, we have
\begin{equation}
    \Lambda(d_1,p) \leq \Psi(d_2, R).
\end{equation}

We then prove that for 
\[
p\in\left[\frac{d_1}{2},\,1-\frac{d_1}{2}\right],
\]
we have
\[{\ \Lambda(d_1,p)=h_2(d_1/2).\ \text{(independent of $p$)}\ }
\]

Note that if we set
\begin{align}
    \alpha&=\left(p-\frac{d_1}{2}\right)
\frac{d_1}{2(1-d_1)}, \label{eq:id_alpha}\\
\beta&=\left(1-\frac{d_1}{2}\right)\left(p-\frac{d_1}{2}\right)\frac{1}{1-d_1}, \label{eq:id_beta}
\end{align}
we get that $\Lambda(d_1,p)\geq h_2(d_1/2)$. It remains to show $\Lambda(d_1,p)\leq h_2(d_1/2)$.
    
Next, one can show that $ \Lambda(d_1,p)$ has the following equivalent form:
\[
\Lambda(d_1,p)=\max_{\text{feasible }(X,Y)} H(X| Y).
\]
where the feasibility constraints are 
$\mathbb P[X=1]=p$, 
$\mathbb P[X\neq Y]\le d_1/2$, and
$\mathbb P[Y=1]\geq p-\frac{d_1}{2}$. 
Define
\[
q \triangleq p-\frac{d_1}{2}+2\alpha.
\]
Then the two weights become $q$ and $1-q$. If we set
\[
\beta = \mathbb P[X=1,Y=1],
\]
then necessarily $\mathbb P[X=1]=p$ is enforced by taking
\[
\mathbb P[X=1,Y=0]=p-\beta.
\]

Now look at the two entropy terms:
\begin{itemize}
\item $q\, h_2\!\left(\frac{\beta}{q}\right)$ equals $\mathbb P[Y=1]\cdot H(X| Y=1)$ because
$\mathbb P[X=1| Y=1]=\beta/q$.
\item $(1-q)\, h_2\!\left(\frac{p-\beta}{1-q}\right)$ equals $\mathbb P[Y=0]\cdot H(X| Y=0)$ because
$\mathbb P[X=1| Y=0]=(p-\beta)/(1-q)$.
\end{itemize}

Hence the maximand is
\[
q\, h_2\!\left(\frac{\beta}{q}\right) + (1-q)\, h_2\!\left(\frac{p-\beta}{1-q}\right)
= H(X| Y),
\]
for some binary pair $(X,Y)$ with $\mathbb P[X=1]=p$. Therefore
\[
\Lambda(d_1,p)=\max_{\text{feasible }(X,Y)} H(X| Y).
\]
The constraints are
$0 \leq \alpha $
which reduces to
$$q\geq p-\frac{d_1}{2}.$$
Also,
$$p-\frac{d_1}{2}+\alpha \leq \beta$$
which reduces to $\mathbb P[X\neq Y]\le \frac{d_1}{2}$, as with the above joint pmf,
\begin{align*}
    &\mathbb P[X\neq Y]
= \mathbb P[X=1,Y=0]+\mathbb P[X=0,Y=1]
\\&= (p-\beta) + (q-\beta)
= p+q-2\beta.
\end{align*}

Rewrite the $\beta$-lower bound using $q=p-\frac{d_1}{2}+2\alpha$. From $\mathcal S_2(\alpha)$,
\[
\beta \ge p-\frac{d_1}{2}+\alpha
= \frac{p+q-\frac{d_1}{2}}{2}.
\]
Equivalently,
\[
p+q-2\beta \le \frac{d_1}{2},
\qquad\text{i.e.}\qquad
\mathbb P[X\neq Y]\le \frac{d_1}{2}.
\]
Thus, Theorem \ref{thm::clique_size_linear} is maximizing $H(X| Y)$ over all $(X,Y)$ with fixed $\mathbb P[X=1]=p$ and mismatch at most $d_1/2$.

Next, we show that concavity gives: $H(X| Y)\le h_2(d_1/2)$.
Define the conditional error probability
\[
e(Y)\triangleq \mathbb P[X\neq Y| Y].
\]
Since $X$ is binary, conditioning on $Y=y$ gives
\[
H(X| Y=y)=h_2\!\big(e(y)\big).
\]
Therefore
\[
H(X| Y)=\mathbb E\big[h_2(e(Y))\big].
\]
By concavity of $h_2(\cdot)$ and Jensen's inequality,
\begin{align*}
    &H(X| Y)=\mathbb E[h_2(e(Y))]
\le h_2\!\big(\mathbb E[e(Y)]\big)
= h_2\!\big(\mathbb P[X\neq Y]\big)
\le h_2\!\left(\frac{d_1}{2}\right).
\end{align*}
This bound does not depend on $p$. Hence for every $p$,
\[
\Lambda(d_1,p)\le h_2(d_1/2).\]




\end{proof}
}

\begin{figure*}[h]
    \centering
    \begin{subfigure}[b]{0.48\textwidth}
        \centering
        \includegraphics[width=\linewidth]{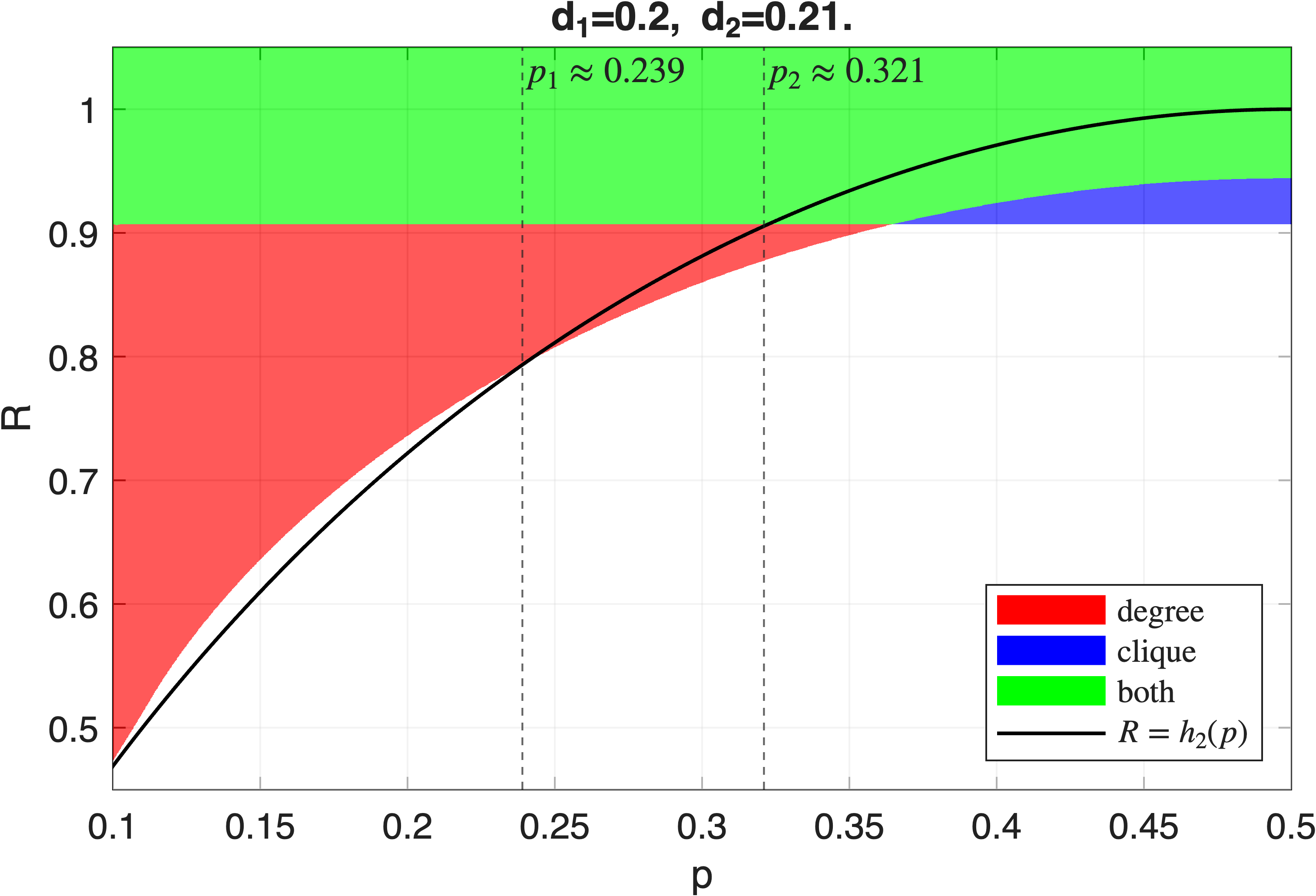}
        \caption{Fixed $d_1 = 0.2 < d_2 = 0.21$.}
        \label{fig:p1}
    \end{subfigure}
    \hfill 
    \begin{subfigure}[b]{0.48\textwidth}
        \centering
        \includegraphics[width=\linewidth]{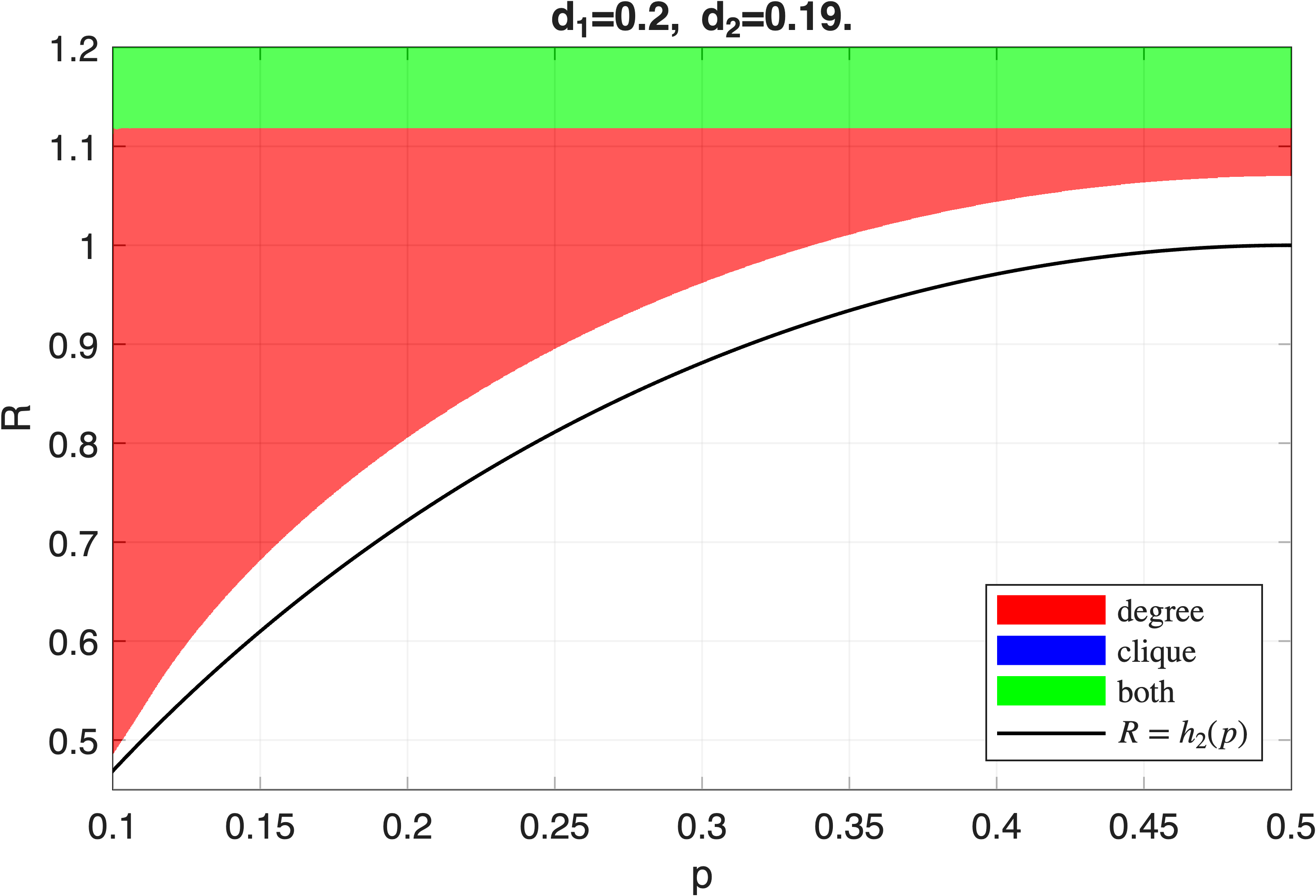}
        \caption{Fixed $d_1 = 0.2 > d_2 = 0.19$.}
        \label{fig:p2}
    \end{subfigure}
        \label{fig:combined_plots}
\end{figure*}

\section{Numerical evaluations}\label{sec::visualizations}

Fig.~\ref{fig:p1} and~\ref{fig:p2} show numerical simulations of our bounds in Theorem \ref{thm::degree_linear} and \ref{thm::clique_size_linear} (in the linear regime) with fixed $d_1$ and $d_2$ and over $(\frac{d_1}{2}\leq p \leq \frac12, R \geq d_2)$ pairs. $(p,R)$ pairs that satisfy both Theorem \ref{thm::degree_linear} and \ref{thm::clique_size_linear} are colored in green. Pairs satisfying Theorem \ref{thm::degree_linear} only are colored in red, and pairs satisfying Theorem \ref{thm::clique_size_linear} only are colored in blue. The black curve is the trivial converse $R \geq H(X) = h_2(p)$.

In Fig.~\ref{fig:p1}, $d_2 = 0.21 > d_1 = 0.2$. The degree bound in Theorem \ref{thm::degree_linear} is dominated by $R \geq h_2(p)$ after $p_1 \approx 0.239$. The clique size bound in Theorem \ref{thm::clique_size_linear} is dominated by $R \geq h_2(p)$ after $p_2 \approx 0.321$. In Fig.~\ref{fig:p2}, $d_2 = 0.19 < d_1 = 0.2$. The clique size bound in Theorem \ref{thm::clique_size_linear} dominates.

\section{Future work}\label{sec::discussions}
One may be able to use bounds on the independence number and the no-homomorphism lemma to obtain other converse bounds. Bounds on the independence number of $G_n$ relate to the maximal size of \emph{a constant weight code} \cite{best1978bounds, johnson2003new, agrell2002upper, graham2003lower}.

\section*{Appendix}

\subsection{Proof of Proposition \ref{prop:reduction-one-type}}
Take a sequence $\epsilon_n$ such that $\epsilon_n\rightarrow 0$ and $\epsilon_n\sqrt{n}\rightarrow\infty$ as $n$ goes to infinity. We have
\begin{equation}
    \lim_{n\rightarrow \infty}\Pr\big[X^n\in \mathcal T_{\epsilon_n}^{(n)}(X)\big]=1.\label{fact1eq}
\end{equation}
Next, let
$\mathcal{B}_n=\{x^n: g_n(f_n(x^n))=x^n\}.$
Since the error probability of the code vanishes, we have
\begin{align}
  \lim_{n\rightarrow \infty}\Pr[X^n\in\mathcal{B}_n]=1  \label{fact2eq}
\end{align} 
Let $\mathcal{F}_n=\mathcal{B}_n\bigcap \mathcal T_{\epsilon_n}^{(n)}(X)$. 
Using the union bound on the complement of sets, from \eqref{fact1eq} and \eqref{fact2eq}, we get
$  \lim_{n\rightarrow \infty}\Pr[X^n\in\mathcal{F}_n]=1  $.

Note that $\mathcal T_{\epsilon_n}^{(n)}(X)=\bigcup_{\pi_n\in\Pi_n({\epsilon_n}, P_X)}\mathcal T(\pi_n)$ is a disjoint union.
Define the conditional (type-wise) ``\emph{good}'' probabilities $\beta_n(\pi_n)\triangleq \Pr[X^n\in\mathcal F_n| X^n\in\mathcal T(\pi_n)].$
Then, since $\mathcal F_n\subseteq \mathcal T_{\epsilon_n}^{(n)}(X)$, we have
\begin{align*}
&\Pr[X^n\in\mathcal F_n]
\leq\Pr[X^n\in\mathcal F_n| X^n\in \mathcal T_{\epsilon_n}^{(n)}(X)]
\\&=\!\!\sum_{\pi_n\in\Pi_n({\epsilon_n}, P_X)} \!\!\beta_n(\pi_n)\,
\Pr\big[X^n\in\mathcal T(\pi_n)| X^n\in \mathcal T_{\epsilon_n}^{(n)}(X)\big].
\end{align*}
Consequently, $\max_{\pi_n\in\Pi_n({\epsilon_n}, P_X)} \beta_n(\pi_n)\ \ge\
\Pr[X^n\in\mathcal F_n].$
In particular, if $\Pr[X^n\in\mathcal F_n]\to 1$, then there exists a sequence
$\{\pi^*_n: \pi^*_n\in\Pi_n({\epsilon_n}, P_X)\}_{n\ge 1}$ such that
$\beta_n(\pi^*_n)\to 1.$
Let
$\mathcal A_n\triangleq \mathcal F_n\cap \mathcal T(\pi^*_n).$
Since conditioned on $X^n\in\mathcal T(\pi_n)$, the distribution of $X^n$ is uniform over $\mathcal T(\pi_n)$, the condition $\beta_n(\pi^*_n)\rightarrow 1$ yields the desired result.

\subsection{Proof of Lemma \ref{lem:deg-preserve-simple}}

\begin{proof}
Write $e(G[S])$ for the number of edges inside $S$ and $\partial(S)$ for the edge boundary
$\{(u,v)\in E(G):u\in S,v\notin S\}$. Counting edge incidences from vertices in $S$ gives $dM = 2e(G[S]) + |\partial(S)|.$
Each vertex in $V(G)\setminus S$ has degree $d$, hence it is incident to at most $d$ boundary edges,
so $|\partial(S)|\le d(N-M)$. Therefore,
$$
2e(G[S]) \ge dM - d(N-M)=d(2M-N).
$$
Dividing by $M$ gives the average degree lower bound, and $\Delta(G[S])$ is at least the average degree.
The final displayed inequality follows by substituting $M\ge (1-\varepsilon)N$.

Next, let $C$ be a maximum clique in $G$ with $|C|=\omega(G)$, and let $\Gamma=\mathrm{Aut}(G)$ be the automorphisms over $G$.
Choose $\sigma\in\Gamma$ uniformly at random. By vertex transitivity, each vertex $v\in V(G)$
is contained in $\sigma(C)$ with probability $|C|/N=\omega(G)/N$. Hence
\[
\mathbb E_\sigma\big[|\sigma(C)\cap S|\big]
=\sum_{v\in S}\Pr[v\in\sigma(C)]
= M\cdot \frac{\omega(G)}{N}.
\]
Therefore, for some $\sigma_0\in\Gamma$,
$|\sigma_0(C)\cap S|\ge M\omega(G)/N$. Since $\sigma_0(C)$ is a clique and
$\sigma_0(C)\cap S\subseteq S$, it is a clique in $G[S]$, proving the claim.
\end{proof}

\clearpage

\bibliographystyle{IEEEtran}
\bibliography{references}

\end{document}